\documentclass[11pt, a4paper]{article}

\usepackage{verbatim}

\usepackage{ifthen}



\newcommand{\ifndef}[2]{\ifthenelse{\isundefined{#1}}{#2}{}}

\newcommand{\mydef}[2]{\def#1{#2}}

\newcommand{\nospell}[1]{#1}  %

\makeatletter
\newcommand{\myusepackage}[2][]{\@ifpackageloaded{#2}{} %
{\ifthenelse{\equal{}{#1}} {\usepackage{#2}} {\usepackage[#1]{#2}} }}
\makeatother

\myusepackage{amssymb}
\myusepackage{amsmath}  %

\myusepackage{amsthm}  %

{}
\myusepackage{latexsym}
\myusepackage{amsfonts}
\myusepackage{units}    %
\myusepackage{psfrag}   %
\myusepackage{url}   %
\myusepackage{hyphenat}  %
\myusepackage{microtype}   %
{}
{}
\myusepackage[dvips]{graphicx}
\myusepackage[usenames,dvipsnames]{color}
{}
{}
{}

{}
\ifndef{\theorem}{\newtheorem{theorem}{Theorem}}
{}

{}  %
{}
\ifndef{\lemma}{\newtheorem{lemma}[theorem]{Lemma}}
\ifndef{\corollary}{\newtheorem{corollary}[theorem]{Corollary}}
\ifndef{\conjecture}{\newtheorem{conjecture}[theorem]{Conjecture}}
\ifndef{\remark}{\theoremstyle{remark} \newtheorem{remark}[theorem]{Remark}}
\ifndef{\proposition}{\newtheorem{proposition}[theorem]{Proposition}}
\ifndef{\claim}{\newtheorem{claim}[theorem]{Claim}}
\ifndef{\result}{\newtheorem{result}[theorem]{Result}}
\ifndef{\problem}{\newtheorem{problem}[theorem]{Problem}}
{}  %
{}  %

\newtheoremstyle{mydefinition}  %
{\topsep}{\topsep}  %
{\slshape}  %
{}  %
{\bfseries}  %
{.}  %
{ }  %
{}  %

\newtheoremstyle{myremark}  %
{\topsep}{\topsep}  %
{\slshape}  %
{}  %
{\bfseries\slshape}  %
{:}  %
{ }  %
{}  %

\newtheoremstyle{myexample}  %
{\topsep}{\topsep}  %
{\itshape}  %
{}  %
{\slshape}  %
{:}  %
{ }  %
{\ul{\thmname{#1}}}  %

\newtheoremstyle{myclaims}  %
{\topsep}{\topsep}  %
{\slshape}  %
{}  %
{\bfseries\itshape}  %
{}  %
{ }  %
{\thmname{#1}\thmnumber{ \!#2}.}  %

{} %
{} %
{\theoremstyle{myremark}}
{}
\ifndef{\definition}
{\theoremstyle{mydefinition}\newtheorem{definition}[theorem]{Definition}}
\ifndef{\example}
{\theoremstyle{myexample}\newtheorem{example}[theorem]{Example}}
{\theoremstyle{myclaims}

\ifndef{\fact}{\newtheorem{fact}[theorem]{Fact}}
}
{} %
{} %
{} %

{}
\newtheoremstyle{anystatement}{\topsep}{\topsep}{\itshape}{}{\bfseries}{.}{ }{\anystatementname}
{\theoremstyle{anystatement}}
\newcommand{\anystatementname}{}

{}

\newcommand{\AuxNew}[4][]{#2{#3}[1][*]%
{\ifthenelse{\equal{*}{##1}}%
{\Ensuremath{#1{#4}}}%
{\ifthenelse{\equal{b}{##1}}%
{\Ensuremath{\mathbf{#4}}}%
{\ifthenelse{\equal{}{##1}}%
{\IfMathMode{#1{#4}}{#4}}{}}}}}

\newcommand{\newident}[3][*]{\ifthenelse{\equal{*}{#1}}%
{\AuxNew[\mathit]{\newcommand}{#2}{#3}}%
{\mydef{#2}{\Ensuremath{\mathit{#3}}}}}%

\newcommand{\newmat}[3][*]{\ifthenelse{\equal{*}{#1}}%
{\AuxNew{\newcommand}{#2}{#3}}%
{\mydef{#2}{\Ensuremath{#3}}}}%

\newcommand{\providemat}[3][*]{\ifthenelse{\equal{*}{#1}}%
{\AuxNew{\providecommand}{#2}{#3}}%
{\mydef{#2}{\Ensuremath{#3}}}}%

\newcommand{\providematarg}[2]{%
\providecommand{#1}[1][]{\Ensuremath{#2}}}

\newcommand{\newfunction}[2]{%
\newcommand{#1}[2][*]{\ifthenelse{\equal{*}{##1}}%
{\Ensuremath{#2{\left(##2\right)}}}%
{#2(##2)}}%
}

\newcommand{\MyMakeTheoMacros}[3]{
\newcommand{#2}[2][]{\ifthenelse{\equal{}{##1}}
{\begin{#1} ##2 \end{#1}}
{\begin{#1}\label{##1} ##2\end{#1}}}
\newcommand{#3}[3][]{\ifthenelse{\equal{}{##1}}
{\begin{#1}[##2] ##3 \end{#1}}
{\begin{#1}[##2]\label{##1} ##3\end{#1}}}
}

\makeatletter
\newtheorem*{rep@theorem}{\rep@title}
\newcommand{\newreptheorem}[2]{%
\newenvironment{rep#1}[1]{%
\def\rep@title{#2 \ref{##1}}%
\begin{rep@theorem}}%
{\end{rep@theorem}}}
\makeatother

\newcommand{\MyMakeDupTheoMacros}[7]{
\MyMakeTheoMacros{#1}{#2}{#3}
\newreptheorem{#1}{#6}
\newcommand{#4}[3]{
\newcommand{##2}{##3}
\begin{#1}\label{##1} ##2\end{#1}}
\newcommand{#5}[4]{
\newcommand{##2}{##4}
\begin{#1}{\e{##3}}\label{##1} ##2\end{#1}}
\newcommand{#7}[2]{\begin{rep#1}{##1} ##2 \end{rep#1}}
}

\newcommand{\MyMakeRefMacros}[3]{\newcommand{#1}[2][]
{\ifthenelse{\equal{}{##1}}{#2~\ref{##2}}{#3~\ref{##1} and~\ref{##2}}}}

\newcommand{\MyMakeEqRefMacros}[3]{\newcommand{#1}[2][]
{\ifthenelse{\equal{}{##1}}{#2~\eqref{##2}}{#3~\eqref{##1} and~\eqref{##2}}}}

{}  %
\newcommand{\bibentry}[8]{
{}\bibitem[\nospell{#8}]{#1} {\textup #3}.{}
\ifthenelse{\equal{}{#6}}
{\newblock \textrm{#4.} \newblock {\em #5}, #7.}
{\newblock \textrm{#4.} \newblock {\em #5, #6}, #7.}
}

{} %

{} %
\MyMakeDupTheoMacros{lemma}
{\lem}{\nlem}{\lemdup}{\nlemdup}{Lemma}{\lemrep}
{}

\MyMakeRefMacros{\lemref}{Lemma}{Lemmas}

\MyMakeDupTheoMacros{corollary}
{\crl}{\ncrl}{\crldup}{\ncrldup}{Corollary}{\crlrep}

\MyMakeRefMacros{\crlref}{Corollary}{Corollaries}

\MyMakeTheoMacros{proposition}{\prp}{\nprp}

{}
\newtheorem*{prp*}{\e{Proposition}}

{}

\MyMakeRefMacros{\prpref}{Proposition}{Propositions}

\MyMakeDupTheoMacros{my_claim}
{\clm}{\nclm}{\clmdup}{\nclmdup}{Claim}{\clmrep}

\MyMakeRefMacros{\clmref}{Claim}{Claims}

\MyMakeDupTheoMacros{theorem}
{\theo}{\ntheo}{\theodup}{\ntheodup}{Theorem}{\theorep}

\MyMakeRefMacros{\theoref}{Theorem}{Theorems}

\MyMakeTheoMacros{definition}{\defi}{\ndefi}

\MyMakeRefMacros{\defiref}{Definition}{Definitions}

\MyMakeTheoMacros{problem}{\prob}{\nprob}

\MyMakeRefMacros{\probref}{Problem}{Problems}

\MyMakeTheoMacros{conjecture}{\conj}{\nconj}

\MyMakeRefMacros{\conjref}{Conjecture}{Conjectures}

\renewcommand{\qedsymbol}{$\blacksquare$}

\newcommand{\prf}[2][]{\ifthenelse{\equal{}{#1}}%
{\begin{proof}\renewcommand{\qedsymbol}{$\blacksquare$}%
#2 \end{proof}}%
{\begin{proof}[Proof of #1]%
\renewcommand{\qedsymbol}{$\blacksquare_{\mbox{\it{\scriptsize{#1}}}}$}%
#2 \end{proof}\renewcommand{\qedsymbol}{$\blacksquare$}}%
}

\newcommand{\abstr}[1]{\begin{abstract} #1 \end{abstract}}

\newcommand{\sect}[2][]{\ifthenelse{\equal{}{#1}}
{\section{#2}}
{\section{#2}\label{#1}}}

\MyMakeRefMacros{\chref}{Chapter}{Chapters}

\MyMakeRefMacros{\sref}{Section}{Sections}

\MyMakeRefMacros{\ssref}{Subsection}{Subsections}

\MyMakeRefMacros{\sssref}{Subsection}{Subsections}

\MyMakeRefMacros{\figref}{Figure}{Figures}

\newcommand{\IfMathMode}[2]{\ifmmode{#1}\else{#2}\fi}

\newcommand{\Ensuremath}{\ensuremath}

\newcommand{\fbr}[1]{\IfMathMode%
{#1}{$#1$}}                     %

\newcommand{\fnbr}[1]{\mbox{\fbr{#1}}}  %

\newcommand{\fla}[2][*]{\ifthenelse{\equal{}{#1}}{\fbr{#2}}{\fnbr{#2}}}

\newcommand{\mat}[2][]{\ifthenelse{\equal{}{#1}}%
{ \begin{displaymath} #2 \end{displaymath} }%
{ \begin{equation} \label{#1} #2 \end{equation} }%
}

\newcommand{\malabel}[1]{\addtocounter{equation}{1}\tag{\theequation}\label{#1}}
\newcommand{\mal}[2][]{ %
\ifthenelse{\equal{}{#1}} %
{{\begin{align*} #2 \end{align*}}}  %
{\ifthenelse{\equal{P}{#1}}                %
{{\allowdisplaybreaks\begin{align*} #2           %
\end{align*}}} %
{{\begin{align*} \malabel{#1} #2 \end{align*}}}  %
} %
}

\newcommand{\m}{\mat}

\MyMakeEqRefMacros{\equref}{Equation}{Equations}

\MyMakeEqRefMacros{\expref}{Expression}{Expressions}

\MyMakeEqRefMacros{\inequref}{Inequality}{Inequalities}

\newcommand{\bracref}[1]{(\ref{#1})}

\newcommand{\bref}{\bracref}

\providecommand{\middle}{\big}

\newcommand{\h}[2][]{\ifthenelse{\equal{}{#2}}%
{\mathop h_{#1}}%
{\mathop h_{#1}{\left[{#2}\right]}}}
\newcommand{\hh}[3][]{\mathop h_{#1}%
{\left[{#2}\middle|\vphantom{|_1^1}{#3}\right]}}

\newcommand{\KL}[2]{d_{KL}\llp{#1}\middle\|\vphantom{|_1^1}{#2}\rrp}

\newcommand{\I}[3][]{\ifthenelse{\equal{}{#1}}%
{\mathbf{I}{\left[{#2}:\vphantom{|_1^1}{#3}\right]}}%
{\mathbf{I}{\left[{#2}:\vphantom{|_1^1}{#3}\middle|{#1}\right]}}}

\providecommand{\E}[2][]{\ifthenelse{\equal{}{#1}}%
{\mathop{\mathbf{E}}{\left[{#2}\right]}}%
{\mathop{\mathbf{E}}_{#1}{\left[{#2}\right]}}}

\newcommand{\PR}[2][]{\mathop{\mathbf{Pr}}_{#1}{\left[{#2}\right]}}

\newcommand{\ord}[1][]{\nospell{\ifthenelse{\equal{}{#1}}%
{\txt{'th}}%
{\ifthenelse{\equal{1}{#1}}{$1\txt{'st}$}{\ifthenelse{\equal{2}{#1}}{$2\txt{'nd}$}{\ifthenelse{\equal{3}{#1}}{$3\txt{'rd}$}{\fla{#1\txt{'th}}}}}}}}

\newcommand{\fr}[3][*]{%
\ifthenelse{\equal{*}{#1}}       %
{\frac{#2}{#3}}{}%
\ifthenelse{\equal{/}{#1}}       %
{\nicefrac{#2}{#3}}{}%
\ifthenelse{\equal{}{#1}}        %
{\left.#2\middle/#3\right.}{}%
\ifthenelse{\equal{p_}{#1}}      %
{\left.\left(#2\right)\middle/#3\right.}{}%
\ifthenelse{\equal{_p}{#1}}      %
{\left.#2\middle/\left(#3\right)\right.}{}%
\ifthenelse{\equal{pp}{#1}}      %
{\left.\left(#2\right)\middle/\left(#3\right)\right.}{}
}

\def\MySQRT#1#2{   %
\setbox0=\hbox{$#1\sqrt{#2\,}$}\dimen0=\ht0%
\advance\dimen0-0.2\ht0%
\setbox2=\hbox{\vrule height\ht0 depth -\dimen0}%
{\box0\lower0.4pt\box2}}

\newcommand{\set}[2][]{\ifthenelse{\equal{}{#1}}%
{\Ensuremath{\left\{#2\right\}}}%
{\Ensuremath{\left\{#2\middle|\vphantom{|_1^1}#1\right\}}}}

\newfunction{\asO}{O}
\newfunction{\aso}{o}
\newfunction{\asOm}{\Omega}
\newfunction{\astOm}{\tilde \Omega}
\newfunction{\asT}{\Theta}

\mydef{\01}{\set{0,1}}
\mydef{\12}{\set{1,2}}

\newcommand{\txt}[1]{\textrm{#1}}  %

\newcommand{\Cl}{\mathcal}  %

\DeclareMathAlphabet{\lowcal}{OT1}{pzc}{m}{it}

\newmat[]{\llp}{\left(}
\newmat[]{\rrp}{\right)}

\newmat[]{\dt}{\cdot}
\newmat[]{\tm}{\cdot}
\newmat[]{\deq}{\stackrel{\textrm{def}}{=}}

\providemat{\QQ}{\mathbb{Q}}
\providematarg{\NN}{\ifthenelse{\equal{}{#1}}%
{\mathbb{N}}%
{\mathbb{N}_{#1}}}

\newcommand{\ds}[1][]
{\ifthenelse{\equal{}{#1}}{\dots}{#1\dots#1}}
\newmat[]{\dc}{\ds[,]}

\newcommand{\itemi}[2][]{\ifthenelse{\equal{}{#1}}
{\begin{itemize} #2 \end{itemize}}
{\begin{itemize}[#1] #2 \end{itemize}}}

\newcommand{\MyComment}[1]{\ClassWarning{My Macros}{#1}}

\newcommand{\wrt}{w.r.t.\ }	%

\newcommand{\fn}{\footnote}

{} %
\newcommand{\e}{\emph}
{}  %

\providecommand{\ul}[1]{\underline{#1}} %

\MyComment{Look for ...-s}

\MyComment{Spell-check}

{}
{}

\renewcommand{\l}{\left}
\renewcommand{\r}{\right}

\newident{\EQ}{EQ}
\newident{\NE}{NE}

\title{Equality, Revisited}

\date{}

\newcommand{\instDG}{Institute of Mathematics, Academy of Sciences, \v Zitna 25, Praha 1, Czech Republic.}
\newcommand{\thanksDG}{Partially funded by the grant P202/12/G061 of GA \v CR and by RVO:\ 67985840.
Most of this work was done while DG was visiting the Centre for Quantum Technologies at the National University of Singapore, and was partially funded by the Singapore Ministry of Education and the NRF.
}

\newcommand{\instHK}{Division of Mathematical Sciences, Nanyang Technological University, Singapore \& Centre for Quantum Technologies, National University of Singapore, Singapore.
}
\newcommand{\thanksHK}{This work is funded by the Singapore Ministry of Education (partly through the Academic Research Fund Tier 3 MOE2012-T3-1-009) and by the Singapore National Research Foundation.}
\newcommand{\instRB}{Division of Mathematical Sciences, Nanyang Technological University, Singapore.}

{}
\author{Ralph C. Bottesch\thanks{\instRB}\and Dmitry Gavinsky\thanks{\instDG\ \thanksDG}
\and Hartmut Klauck\thanks{\instHK\ \thanksHK}
}
{}

\begin{document}

\maketitle

\thispagestyle{empty}

\abstr{In 1979 Yao published a paper that started the field of communication complexity and asked, in particular, what was the randomised complexity of the Equality function (\EQ) in the Simultaneous Message Passing (SMP) model (for the question to be non-trivial, one must consider the setting of \e{private randomness}).
The tight lower bound \asOm{\sqrt n} was given only in 1996 by Newman and Szegedy.

In this work we develop a new lower bound method for analysing the complexity of \EQ\ in SMP.
Our technique achieves the following:
\itemi{
\item It leads to the tight lower bounds of \asOm{\sqrt n} for both \EQ\ and its negation \NE\ in the \e{non-deterministic} version of quantum-classical SMP, where Merlin is also quantum -- this is the strongest known version of SMP where the complexity of both \EQ\ and \NE\ remain high (previously known techniques seem to be insufficient for this).
\item It provides a unified view of the communication complexity of \EQ\ and \NE, allowing to obtain tight characterisation in all previously studied and a few newly introduced versions of SMP, including all possible combination of either quantum or randomised Alice, Bob and Merlin in the non-deterministic case.
Arguably, it also simplifies the previously known lower bound proofs.
}

Our characterisation also leads to tight trade-offs between the message lengths needed for players Alice, Bob, Merlin, not just the maximum message length among them, and highlights that \NE\ is easier than \EQ\  in the presence of classical proofs, whereas the problems have (roughly) the same complexity when a quantum proof is present.

We also construct new protocols for \EQ\ and \NE\ that achieve optimal trade-offs in the ``asymmetric'' scenarios when the (qu)bits from Merlin are either cheaper or more expensive than those from the trusted parties.
Along the way, we give tight analysis of a new primitive, where a honest classical and an untrusted quantum parties help a third party to obtain an approximate copy of a quantum state.
}

\newpage

\setcounter{page}{1}

\sect[intro]{Introduction}

The Equality function (\EQ) and the Simultaneous Message Passing model (SMP) are among the longest-studied objects in communication complexity:
When in 1979 Yao published his seminal paper~\cite{Y79_So_Co} that introduced communication complexity, \EQ\ was repeatedly used as an example (under the name of ``the identification function''), SMP was introduced (referred to as ``$1\to3\gets2$''), and determining the SMP complexity of \EQ\ was posed as an open problem.

Being one of the weakest models that has been studied in communication complexity, SMP is, probably, the most suitable for studying \EQ.
While in several natural variants of SMP the complexity of \EQ\ varies from constant to \asOm{\sqrt n}, in the most of more powerful setups \EQ\ can be solved by very efficient protocols of at most logarithmic cost.\fn
{The communication complexity of \EQ\ becomes $n$ in the most of \e{deterministic} models, but those results are usually trivial and we do not consider the deterministic setup in this work (except for one special case where a ``semi-deterministic'' protocol has complexity \asO{\sqrt n} -- that situation will be analysed in one of our lower bound proofs).}
That happens due to the fact that \EQ\ becomes very easy (even for SMP) in the presence of shared randomness, and virtually all commonly studied stronger models can emulate shared randomness at the cost of at most \asO{\log n} additional bits of communication.

The standard SMP setting involves three participants:\ the \e{players} Alice and Bob, and the \e{referee}; each of them can use private randomness.
The input consists of two parts, one is given to Alice and the other to Bob; upon receiving their parts, the players send one message each to the referee; upon receiving the messages, the referee outputs the answer.
Depending on the considered model, each player can be deterministic, randomised or quantum.
In the non-deterministic regime, there is a third player called Merlin, who knows both Alice's and Bob's parts of the input, and who sends his message to the referee, but is a dishonest person whose goal is to convince the referee to accept.\fn
{Note that Merlin's message is only seen by the referee, and not by Alice and Bob.
Letting the players receive messages from Merlin prior to sending their own messages would contradict the ``simultaneous flavour'' of the SMP model.
Practically, that would make \NE\ trivial; while the case of \EQ\ is less obvious, we believe that the techniques developed in this work would be useful there as well.}
All players are assumed to be computationally unlimited, and the cost of a protocol is the (maximum) number of (qu)bits that the players send to the referee.
Unless stated otherwise, we consider the \e{worst-case setting}, where a protocol should give correct answer on every possible input with probability at least $2/3$.

Yao's question about the SMP complexity of \EQ\ was answered seventeen years later, in 1996, by Ambainis~\cite{A96_Co_Co}, who gave a protocol of cost \asO{\sqrt n}, and by Newman and Szegedy~\cite{NS96_Pub}, giving the matching lower bound \asOm{\sqrt n}. Babai and Kimmel~\cite{EQalso} showed the same lower bound shortly afterwards using a simpler argument.
The above results address the ``basic'' version of SMP, where all the participants are randomised and no shared randomness is allowed (recall that otherwise SMP becomes ``too powerful'' for \EQ).

In 2001, Buhrman, Cleve, Watrous and de Wolf~\cite{BCWW01_Qu_Fi} considered the version of SMP with quantum players and gave a very efficient and surprising protocol solving \EQ\ at cost \asO{\log n}, and showed its optimality.
In 2008, Gavinsky, Regev and de Wolf~\cite{GRW08_Si_Co} studied the ``quantum-classical'' version of SMP, where only one of the players could send a quantum message, and they showed that the complexity of \EQ\ in that model was \asOm{\sqrt{n/\log n}} (which was, tight up to the multiplicative factor $\sqrt{\log n}$ by \cite{NS96_Pub}, and was improved to \asOm{\sqrt n} in \cite{KlauckPodder14}) .

{}
\subsection{Our results}
{}

In this work we revisit the question about the SMP complexity of the Equality function \EQ\ and its negation \NE\ (the two cases are different in the context of non-deterministic models).
We give a complete characterisation of the complexity of \EQ\ and \NE\ in a number of new non-deterministic SMP models corresponding to all possible combinations of either classical or quantum Alice, Bob, and Merlin.
Moreover, our characterisation also covers the ``asymmetric'' scenarios when the (qu)bits from (untrusted) Merlin are either cheaper or more expensive than those from the trusted parties.

Let us denote the {\it type} of a non-deterministic SMP model by three letters, like QRQ, RRQ, RRR etc., where the letter in the corresponding position determines whether, respectively, Alice's, Bob's or Merlin's message is quantum or randomised.
We will say that a protocol is ``$(a,b,m)$'' if Alice, Bob and Merlin sends at most $a$, $b$ and $m$ (qu)bits, respectively.

\begin{table}
\begin{center}
\begin{tabular}{|l|l|l|l|l|}
\hline
Type & Task & Lower bound & Non-trivial upper bound & Tight bound \\
of SMP &  & (assuming $ab=\aso n$) & (skipping $O$) & on $a+b+m$ \\
\hline
RRR
& \EQ & $m=\asOm n$ & $(\log n, \log n,n)$ & $\asT{\sqrt n}$\\
\cline{2-5}
& \NE & $m=\asOm{\fr[/]na+\fr[/]nb}$ & $(a,a,\fr[/]na)$ & $\asT{\sqrt n}$\\
\hline
QRR
& \EQ & $m=\asOm n$ & $(\log n, \log n,n)$ & $\asT{\sqrt n}$\\
\cline{2-5}
& \NE & $m=\asOm{\fr[/]na+\fr[/]nb}$ & $(a,a,\fr[/]na)$ & $\asT{\sqrt n}$\\
\hline
RRQ
& \EQ & $m=\asOm{\fr[/]na+\fr[/]nb}$ & $(a\log a, a\log a,\fr[/]na\cdot\log n)$ & $\asT{\sqrt n}$\\
\cline{2-5}
& \NE & $m=\asOm{\fr[/]na+\fr[/]nb}$ & $(a,a,\fr[/]na)$ & $\asT{\sqrt n}$\\
\hline
QRQ
& \EQ & $m=\asOm{\min\{\fr[/]na,\fr[/]nb\}}$ & $(\log n, b\log b,\fr[/]nb\cdot\log n)$ & $\asT{\sqrt n}$\\
\cline{2-5}
& \NE & $m=\asOm{\min\{\fr[/]na,\fr[/]nb\}}$ & $(\log n,b\log b,\fr[/]nb\cdot\log n)$ & $\asT{\sqrt n}$\\
\hline
\end{tabular}
\caption{Summary of results: complexity of $(a,b,m)$ protocols.}
\label{TA:res}
\end{center}
\end{table}
\vspace{-.2cm}
Our results for all possible types of protocols are summarised in Table~\ref{TA:res}.
As it was shown in~\cite{A96_Co_Co} that both \EQ\ and \NE\ can be computed without any message from Merlin if $ab=const\cdot n$ and $a,b\geq\log n$, we present our hardness results via lower bounds on $m$ as a function of of $a$ and $b$, assuming $ab=o(n)$.
Note that we are closing the gap left by~\cite{GRW08_Si_Co}, as our results about non-deterministic SMP complexity of \EQ\ imply that its quantum-classical complexity is also in \asOm{\sqrt n}. This result has also recently been obtained by Klauck and Podder \cite{KlauckPodder14}.

One of the key ingredients in our lower bound method is a tight analysis of a new communication primitive that we call \e{``One out of two''}, which might be of independent interest:
\itemi{
\item Alice receives $X_1$ and $X_2$, Bob receives $Y$.
\item It is promised that either $X_1=Y$ or $X_2=Y$.
\item The referee has to distinguish the two cases.
}
The problem closely resembles \EQ, but can be shown to be easier in some situations:
We show that its quantum-classical SMP complexity is \asOm{\sqrt n}, and here the interesting case is when Alice's message is quantum (the case of ``quantum Bob'' can be handled relatively easily by previously known techniques).
A new method developed for its analysis can be viewed as the main technical contribution of this work.

In all cases we are showing tightness of our bounds even for arbitrary trade-off between $a$, $b$ and $m$ (Table~\ref{TA:res}).
For that we demonstrate two protocols:\ one for \NE\ in RRR, and one for \EQ\ in QRQ and RRQ.
For the latter, we combine the protocol of~\cite{BCWW01_Qu_Fi} for \EQ\ with a new primitive in quantum communication, which might be of independent interest:

Assume that both Alice and Merlin know the classical description of a quantum state on $\log n$ qubits.
Another player, the referee, wants to obtain one approximate copy of this state.
Alice is trusted, but can send only classical information; Merlin can send quantum messages, but not be trusted.
For which message lengths $a$ and $m$ from, respectively, Alice and Merlin can this be achieved?
We call this problem \e{``Untrusted quantum state transfer''} and give a tight analysis, up to log-factors.
We show that, essentially, $am=\tilde{\Theta}(n)$ is both enough and required.

\section{Organisation}
The remainder of the paper is organised as follows: Section \ref{prelim_sect} contains the preliminaries for the rest of the paper. In Section \ref{s_method_sect} we present our new lower bound technique, which we further extend in Section \ref{exten_sect}. Our main results, lower bounds and matching upper bounds for \NE\ and \EQ\ in different versions of the SMP setting, are given in Sections \ref{lower_sect} and \ref{upper_sect}. Section \ref{uqst_sect} contains a detailed analysis of the ``Untrusted Quantum State Transfer'' problem, and based on this we sketch an RRQ protocol in Section \ref{rrq_protocol_sect}. Finally, unrelated to our results for \NE\ and \EQ, we give an RRR protocol for the Disjointness problem (Section \ref{disj_sect}).

\section{Preliminaries}\label{prelim_sect}
\subsection{Definitions of models}

We assume familiarity with communication complexity,
referring to~\cite{kushilevitz&nisan:cc} for more details about classical communication complexity
and~\cite{wolf:qccsurvey} for quantum communication complexity (for information about
the quantum model beyond what's provided in~\cite{wolf:qccsurvey}, see~\cite{nielsen&chuang:qc}).

In an SMP protocol there are players Alice, Bob, Merlin, who send messages to Referee. None of the parties share any public randomness or entanglement. Alice, Bob, Merlin send messages to Referee, who computes the function value.
Alice, Bob, Merlin never communicate with each other. We always allow private randomness to be used.

We will refer to SMP protocols by shorthands like QQ, QRQ, etc., which denote who is sending what type of message. E.g.~QR refers to a protocol in which Alice sends a quantum message and Bob a classical message. QRQ refers to the situation where Alice and Merlin send a quantum message, Bob a classical message. We do not consider the situation where Merlin sends a message, but either Alice or Bob do not, i.e., QQ, QR, RR refer to Alice and Bob alone.

An $(a,b,m)$-protocol is one in which the message lengths of player Alice, Bob, Merlin are $O(a), O(b), O(m)$ respectively.
So for example, we will consider QRQ $(\log n, \sqrt n\log n,\sqrt n\log n)$-protocols.

\subsection{Quantum states and the swap-test}

We refer to~\cite{nielsen&chuang:qc} for general quantum background.

The {\em trace norm} on matrices $A$ is defined by $||A||_t\deq \rm{Tr}\sqrt{AA^\dagger}$.
We will make use of the {\em trace distance} between two mixed quantum states $\rho,\sigma$, which is defined as $||\ \rho-\sigma||_t$.

Another popular measure of distinguishability between quantum states is the {\em fidelity}, which is defined by $F(\rho,\sigma)=||\sqrt\rho\sqrt\sigma||_t$.
Fuchs and van der Graaf~\cite{FuchsGraaf} show the following relation between the two.
\begin{fact}
For any two mixed states $\rho,\sigma$ we have
\[1-F(\rho,\sigma)\leq \frac{1}{2}||\rho-\sigma||_t\leq \sqrt{1-F^2(\rho,\sigma)}.\]
\end{fact}

We need the main result of~\cite{BCWW01_Qu_Fi} regarding the so-called swap-test. The swap-test takes two quantum states and decides whether they are the same or not. While the original analysis is given only for pairs of pure states, an analysis in~\cite{multmerlion} can be rewritten to establish the following:

\begin{fact}
Given a pure state $|\phi\rangle$ with density matrix $\phi$ and a mixed state $\rho$, the swap-test will accept with probability exactly \[1/2+\frac{F^2(\phi,\rho)}{2}.\]
\end{fact}

{}
\subsection{The new lower bound technique}
{}\label{s_method_sect}

Let us illustrate the idea of our method by applying it in the usual (randomised) SMP regime.
Given a short ``candidate protocol'' we prove that it cannot solve \EQ, as follows:
Let $A(X)$ be the mapping (possibly, randomised) that Alice uses to create her message, and similarly for $B(Y)$.
The goal of the referee is to use the messages $A(X)$ and $B(Y)$ in order to distinguish between the cases $X=Y$ and $X\neq Y$.

The new method is is based on constructing a distribution $\mu$ on $\{0,1\}^n$ of high entropy, such that $X=Y$ chosen according to $\mu$ leads to messages $A(X),B(Y)$ with small mutual information ${\bf I}[A(X):B(Y)]$.
This implies that the referee's state on input $X=Y$ is close to his state when $X$ and $Y$ are chosen from the same distribution $\mu$ each, but independently.
Since independent $X$ and $Y$ satisfy $X\neq Y$ with high probability in the latter case (recall that $\mu$ has high entropy) and the referee cannot distinguish the two situations, he cannot decide \EQ.

Suppose Alice sends $k$ bits and Bob $l$ bits. Assume we create $100\cdot k$ independent samples $B_i$ of Bob's message (for the same input $X$, but using different random choices each time).
With respect to uniformly-random $X$,
\begin{eqnarray*}
k&\geq& {\bf I}[A(X):B_1(X),\ldots,B_{100k}(X)]\\
&=&\sum_i {\bf I}[A(X):B_i(X)|B_1(X),\ldots,B_{i-1}(X)].
\end{eqnarray*}

Hence, for some $i_0$, ${\bf I}[A(X):B_{i_0}(X)|B_1(X),\ldots,B_{i_0-1}(X)]\leq 0.01$.
We let $\mu$ be the distribution of $X$, conditioned on some fixed ``typical'' values $b_1\dc b_{i_0-1}$ of $B_1(X)\dc B_{i_0-1}(X)$.
Then the entropy of $\mu$ is at least $n-O(kl)$ and the information between Alice's and Bob's messages is small when $X=Y\sim\mu$ -- as required for our lower bound method to work.

Let us compare our technique with the one used in~\cite{EQalso} (and later in~\cite{GRW08_Si_Co}).
There the main idea was to amplify the success probability by making one player's message longer, until that message can be made deterministic -- thus forcing communication $\Omega(n)$ from that player's side.
An advantage of this was that it could be applied to any function; the main disadvantage for the special case of \EQ\ was that this approach was ``too demanding'', and could not work in the stronger modifications of SMP.
\\\\
Before going into the details of the proof, let us first look at the limitations of the previous techniques.
For that we can use the ``One out of two'' problem that was defined in the Introduction.
Note that this problem naturally arises in the situation where a prover is around: instead of providing a (quantum) proof that $X\neq Y$ to the referee, Merlin may just send a string $Z$ to any of the players (say, Alice), that is supposed to be the input string of the other player. Alice and Bob should cause the referee to accept if $X\neq Y$ and $Z=Y$, and the referee must certainly reject if $X=Y$, no matter what $Z$ is.

From the communication complexity perspective, the ``One out of two'' is easier than \EQ:
On the one hand, any protocol for \EQ\ can be used to solve ``One out of two''.
On the other hand, in the ``semi-deterministic'' model where Alice's message is deterministic and Bob's message is randomised, the complexity of \EQ\ is known to be \asOm{n}, while ``One out of two'' can be solved by the following protocol\fn
{This protocol is an adaptation of the private-coin protocol for \EQ\ given by Ambainis in~\cite{A96_Co_Co}.}
of cost \asO{\sqrt n} that gives the correct answer with probability at least $2/3$:
\itemi{
\item Alice and Bob fix an error-correcting code $C:\01^n\to\01^N$, where $N\leq\asO n$ and $C(x)$ differs from $C(y)$ on at least $N/3$ positions for every $x\ne y$.
Additionally, let $N=k^2$ for some $k\in\NN$.
\item Viewing the code-words of $C$ as $k\times k$ Boolean matrices, Bob chooses a uniformly random $i\in[k]$, and Alice chooses the smallest $j\in[k]$ such that $C(X_1)$ and $C(X_2)$ differ on at least $k/3$ positions of the \ord[j] row.
\item Bob sends $(i,b_i)$ and Alice sends $(j,a_j^{(1)},a_j^{(2)})$ to the referee, where $b_i$ is the \ord[i] column of $C(Y)$, and $a_j^{(1)}$ and $a_j^{(2)}$ are the \ord[j] rows of $C(X_1)$ and $C(X_2)$, respectively.
\item If the \ord[i] entries of $a_j^{(1)}$ and $a_j^{(2)}$ are different, then the referee answers ``$X_t=Y$'', where $t\in\12$ is such that the \ord[j] entry of $b_i$ equals the \ord[i] entry of $a_j^{(t)}$; otherwise, the referee outputs one of the two possible answers uniformly at random.
}

In all other variants of SMP considered in this work, the complexities of \EQ\ and ``One out of two'' are asymptotically equal.
Nevertheless, the above protocol can be blamed for the fact that a new lower-bound method is required in order to get the tight lower bound \asOm{\sqrt n} in the model of quantum-classical SMP (i.e., Alice is quantum and Bob is randomised).
The only previously known lower bound technique that we found applicable in this situation is the one used in~\cite{GRW08_Si_Co}, namely\itemi{
\item Let $c$ denote the communication cost of a classical-quantum SMP protocol that solves the problem.
Use the technique of Aaronson~\cite{A04_Lim} to construct another protocol that solves the same problem, where the quantum message of the original protocol is replaced by a deterministic one of length \asO{c^2\log c}.
\item Conclude that $c=\asOm{\sqrt{d/\log d}}$, where $d$ is the ``deterministic-randomised'' complexity of the communication problem under consideration.}
Obviously, in the case of the ``One out of two'' problem this only gives a bound \astOm{n^{1/4}}.
The argument can be improved (by considering the ``asymmetric'' version of the deterministic-randomised-deterministic model) to give a bound \astOm{n^{1/3}}, but that seems to be the best possible. An \astOm{n^{1/3}} bound can also be derived directly from the ``Merlin removal'' theorem by Aaronson~\cite{Aaronson06qma/qpoly} combined with the above idea of making the quantum player deterministic.

In this work we circumvent this obstacle posed by the fact that the ``One out of two'' problem does not need a long message from a deterministic Alice by developing a new lower bound method for \EQ\ in the SMP model.
The new argument is very robust -- in particular, in its basic form it can be used to show the \asOm{\sqrt n} lower bound for \EQ\ in both classical and quantum-classical versions of SMP, and its more involved version can be used to obtain the tight bound for the ``One out of two'' problem, as well as for \EQ\ in various non-deterministic SMP settings.

To introduce our method, let us use it to give a new proof that \EQ\ has complexity \asOm{\sqrt n} in the basic SMP model, where both Alice and Bob send a randomised classical message to the referee.

Let $\Cl{P}$ be a protocol, where Alice sends the message $\alpha(X,R_A)$ when her input is $X$ and the random string is $R_A$, and Bob sends $\beta(Y,R_B)$ when his input is $Y$ and the random string is $R_B$.
Let $a$ and $b$ denote the lengths of Alice's and Bob's messages, respectively. First we let both $X$ and $Y$ come from the uniform distribution on $\{0,1\}^{n}$, modulo $X=Y$.

We are interested in the scenario where upon receiving $Y$, Bob produces a sequence of messages $\beta(Y,R_1),\beta(Y,R_2),\dots$ for mutually independent random values $R_i$.
For all $i\in\NN$, let $B_i$ be the random variable that takes the value $\beta(Y,R_i)$, and let $A$ take the value $\alpha(X,R_A)$ -- i.e., $B_i$ is the \ord[i] element in the sequence of $B$'s messages and $A$ is Alice's message.
Define:
\m{p_i\deq\PR[B_1'\dc B_{i-1}']
{\I[B_1=B_1'\dc B_{i-1}=B_{i-1}']{B_i}A>\gamma_0},}
where $\gamma_0$ is a small constant, and the variables $B_1'\dc B_{i-1}'$ are distributed like $B_1\dc B_{i-1}$, respectively.

Observe that
\m{\I[B_1\dc B_{i-1}]{B_i}A>p_i\tm\gamma_0,}
and therefore by the chain rule,
\m{\I{B_1\dc B_i}A>\gamma_0\tm\sum_{j=1}^ip_i.}
On the other hand, for every $i\in\NN$ it holds that
\m{\I{B_1\dc B_i}A\le\h A\le a,}
and so,
\m{\gamma_0\tm\sum_{j=1}^\infty p_i\le a.}
In particular, there exists $t_0=\asO{a}$, such that
\m[m_t0]{\I[B_1=B_1'\dc B_{t_0}=B_{t_0}']{B_{t_0+1}}A<\gamma_0}
holds with probability at least $2/3$ \wrt randomly chosen $B_1'\dc B_{t_0}'$.

To conclude, we need the following technical claim:
\clm[c_Eh]{Let $X$ be a random variable that is uniformly distributed over $\Cl X$, and let $E$ be a Boolean random variable with $\E E>0$.
Then
\m{\log\l(\fr1{\E E}\r)>\log(|\Cl X|)-\hh X{E=1}-3.}}

\prf[\clmref{c_Eh}]{Let $\nu\deq\log(|\Cl X|)$, then
\m{\nu=\h X\le\nu\tm(1-\E E)+\hh X{E=1}\tm\E E+\h E,}
and therefore,
\m[m_Eh1]{\E E\tm(\nu-\hh X{E=1})\le\h E.}
W.l.g.\ we may assume that $\hh X{E=1}<\nu-3$, and so,
\m[m_Eh2]{\E E<\fr13.}

From \bref{m_Eh1} and \bref{m_Eh2},
\mal{\E E\tm(\nu-\hh X{E=1})&\le\h E\\
&\le\E E\tm\log\l(\fr1{\E E}\r)+\log\l(\fr1{1-\E E}\r)\\
&\le\E E\tm\log\l(\fr1{\E E}\r)+\log\l(1+2\E E\r)\\
&<\E E\tm\l(\log\l(\fr1{\E E}\r)+3\r),}
and the result follows.}

It follows that for every $\beta_1\dc \beta_{t_0}$,
\m{\PR{B_1=\beta_1\dc B_{t_0}=\beta_{t_0}}<\fr1{2^{n-\hh Y{B_1=\beta_1\dc B_{t_0}=\beta_{t_0}}-3}},}
and therefore,
\m[m_t1]{\hh Y{B_1=B_1'\dc B_{t_0}=B_{t_0}'}\ge n/2-3}
holds with probability at least $1-2^{t_0\tm b-n/2}$ \wrt randomly chosen $B_1'\dc B_{t_0}'$.

Now assume towards contradiction that $a\tm b=\aso n$.
Then \bref{m_t0} and \bref{m_t1} simultaneously hold with probability $2/3-\aso1$; let $\beta_1\dc \beta_{t_0}$ be any values, such that both the condition hold when $B_1=\beta_1\dc B_{t_0}=\beta_{t_0}$.
Denote by $\mu_0$ the distribution of $Y$ conditioned on ``$B_1=\beta_1\dc B_{t_0}=\beta_{t_0}$''.

Let us consider the behaviour of $\Cl{P}$ when $X\sim\mu_0$ and $Y\sim\mu_0$, either independently or modulo the condition ``$X=Y$''.
If $X$ and $Y$ are independent, then the messages received by the referee are also mutually independent.
If $X=Y$, then the messages no longer have to be independent, even though their marginal distributions are the same.
Denote by $\sigma_A$ the distribution of Alice's message, by $\sigma_B$ that of Bob's message, and by $\sigma_{AB}$ the joint distribution of the two messages when $X=Y$ (note that when $X$ and $Y$ are independent, their joint distribution is $\sigma_A\times\sigma_B$).

Note that when $X$ and $Y$ are independent, the right answer to \EQ\ should be ``$X\ne Y$'' with probability $1-\aso1$ -- this follows from \bref{m_t1}.
Therefore, up to the additive \aso1-factor, the referee's ability to answer correctly equals its ability to distinguish between the distributions $\sigma_{AB}$ and $\sigma_A\times\sigma_B$ of the messages coming from the players.
From \bref{m_t0},
\m{\gamma_0>\I[(A,B)\sim\sigma_{AB}]AB
=\KL{\sigma_{AB}}{\sigma_A\times\sigma_B}
\ge\fr2{\ln 2}\l(d(\sigma_{AB},\sigma_A\times\sigma_B)\r)^2 ,}
where $\KL\dt\dt$ and $d(\dt,\dt)$ denote, respectively, the Kullback-Leibler divergence and the total variation distance, and the last inequality is Pinsker's inequality.
As $\gamma_0$ can be made arbitrarily small, the protocol $\Cl{P}$ makes an error with probability $1/2-\aso1$ (unless $a\tm b=\asOm n$), which completes our proof.

\section{Extensions}\label{exten_sect}
\subsection{QR protocols}

Note that in the proof sketched above we never utilise the classical nature of Alice's message.
The reason why we need Bob to be classical
is that in order to construct the distribution $\mu$ we interpret
$\I[B_1=B_1'\dc B_{t_0}=B_{t_0}']{B_{t_0+1}}A$ as an expectation over fixing the random variables in the condition, something that is impossible in the quantum case.
Alice's message is not required to be classical, and the whole argument goes through unchanged for QR protocols,
where we use the quantum version of Pinsker's inequality (\cite{qpinsker}, see also~\cite{kntz:qinteraction}) in the last step.

Another generalisation is to extend the lower bound to a model, where Alice and the referee share entanglement (no other two players share entanglement). Without loss of generality (at the expense of a factor of two in the message length) in this case Alice can replace her quantum message by a classical message through teleportation. Nevertheless we can argue that the referee's part of the entangled state plus the message from Alice together carry little information about Alice's input, and leave the remaining argument unchanged. So even in this potentially stronger model we obtain the same $\Omega(\sqrt n)$ lower bound.

\subsection{The \e{``One out of two''} problem}

We are now going to extend the lower bound to the ``One out of two'' problem described in Section~\ref{intro}.
Recall that now Alice (the quantum player) has two inputs $X_1,X_2$, and the task is to decide which one is equal to Bob's input $Y$.
To show that this problem is hard, we will assume, towards contradiction, that the protocol is short and construct a distribution $\rho_1$ on triples $X_1,X_2,Y$, where $X_1=Y$ and $X_2$ is independently distributed like the marginal of $\rho_{1}$ on $Y$, and the messages of Alice and Bob (in the given protocol) are almost independent.
Let $\rho_2$ be a modification of $\rho_1$ with the roles of $X_1$ and $X_2$ reverses, then by a variation of the previous argument, the referee cannot distinguish between these two distributions of input.
Like before, the lower bound will follow by noticing that the two distributions are (mostly) supported on input values leading to the opposite answers to the ``One out of two'' problem.

Before defining the distributions on triples of binary strings, we first define a distribution on binary strings. For any $x_{0}\in\{0,1\}^{n}$, let $B_i$ be random variables taking values $\beta(x_0,R_i)$, for random strings $R_i$, with $i$ running from 1 to some value $t$. Choose $t_0$ uniformly at random between 1 and $t$. We then  take $\mu$ as the distribution obtained by restricting the uniform distribution on $\{0,1\}^{n}$ to the event $B_{1}=\beta_{1},\ldots,B_{t_{0}}=\beta_{t_{0}}$, for some $\beta_{1},\ldots,\beta_{t_{0}}\in\{0,1\}^{n}$ (note that this distribution depends on $t_{0}$ and the $\beta_{i}$'s, for which suitable values will be chosen later).

Now let $X_{0}$, $X$ and $Z$ be pairwise independent random variables following the distribution $\mu$. Let $\rho_{1}$ be the distribution of the vector random variable $(X_{1},X_{2},Y)$ when $X_{1}=X_{0}$, $X_{2}=X$ and $Y=X_{0}$. Similarly, let $\rho_{2}$ be the distribution of $(X_{1},X_{2},Y)$ when $X_{1}=X$, $X_{2}=X_{0}$ and $Y=X_{0}$. And finally, let $\rho$ be the distribution of $(X_{1},X_{2},Y)$ when $X_{1}=Z$, $X_{2}=X$ and $Y=X_{0}$. In the case of $\rho$, we simply see three independent samples under $\mu$, meaning that the three random variables will usually have pairwise different values, and hence inputs drawn according to $\rho$ will almost certainly fall outside of the promise of the problem we study. In the case of $\rho_2$ and $\rho_1$, two of the three strings are copies (drawn from $\mu$), while the third is independent but also from $\mu$.

Our goal now is to show that if $\beta(x_{0},R_{i})$ represents the message produced by Bob on input $x_{0}$ using the random string $R_{i}$, then for some choice of suitable values for $t_{0}$ and the $\beta_{i}$'s, the messages produced by Alice and Bob on $\rho_1$ and on $\rho$ are very similar, and that by symmetry the same holds for the messages on $\rho_2$ and on $\rho$. This implies that the referee can not distinguish between $\rho_1$ and $\rho_2$ (which should almost certainly lead to different outputs), hence the error is close to 1/2.
\\\\
We now fill in the details of the above proof outline. In what follows we take $X_{0}$ and $X$ as random variables following the uniform distribution on $\{0,1\}^{n}$, but always condition on events of the type $B_1=\beta_1\ldots,B_k=\beta_k$, thereby obtaining variables following distributions that are similar to $\mu$ (and later get $\mu$ by fixing suitable values for $t_{0}$ and the $\beta_{i}$'s).

We claim that, with high probability, the message distribution on $\rho_1$ and on $\rho$ are close to each other.

\m{p_i\deq\PR[B_1'\dc B_{i-1}']
{\I[B_1=B_1'\dc B_{i-1}=B_{i-1}']{B_i(X_0)}{A(X_0,X)}>\gamma_0},}
where $\gamma_0$ is a small constant, and the variables $B_1'\dc B_{i-1}'$ are distributed like $B_1\dc B_{i-1}$, respectively. The dependence on $X,X_0$ is highlighted. Note that the probability is only over the $B_j'$, while $X_0,X$ are still unfixed random variables.

Observe that
\m{\I[B_1\dc B_{i-1}]{B_i(X_0)}{A(X_0,X)}>p_i\tm\gamma_0,}
and therefore by the chain rule,
\m{\I{B_1(X_0)\dc B_i(X_0)}{A(X_0,X)}>\gamma_0\tm\sum_{j=1}^ip_i.}
On the other hand, for every $i\in\NN$ it holds that
\m{\I{B_1(X_0)\dc B_i(X_0)}{A(X_0,X)}\le\h A\le a,}
and so,
\m{\gamma_0\tm\sum_{j=1}^\infty p_i\le a.}
In particular, for most $t_0\leq 100a$
\m[m_t2]{\I[B_1=B_1'\dc B_{t_0}=B_{t_0}']{B_{t_0+1}(X_0)}{A(X_0,X)}<\gamma_0}
holds with probability at least $2/3$ \wrt randomly chosen $B_1'\dc B_{t_0}'$.

Again, for every $\beta_1\dc \beta_{t_0}$,
\m{\PR{B_1=\beta_1\dc B_{t_0}=\beta_{t_0}}<\fr1{2^{n-\hh Y{B_1=\beta_1\dc B_{t_0}=\beta_{t_0}}-3}},}
and therefore,
\m[m_t3]{\hh Y{B_1=B_1'\dc B_{t_0}=B_{t_0}'}\ge n/2-3}
holds with probability at least $1-2^{t_0\tm b-n/2}$ \wrt randomly chosen $B_1'\dc B_{t_0}'$.

Now assume towards contradiction that $a\tm b=\aso n$.
Then \bref{m_t2} and \bref{m_t3} simultaneously hold with probability $2/3-\aso1$; let $\beta_1\dc \beta_{t_0}$ be any values, such that both the condition hold when $B_1=\beta_1\dc B_{t_0}=\beta_{t_0}$.

Let us consider the behaviour of the protocol $\Cl{P}$ on $\rho_1$ and on $\rho$. Denote the message state of Alice and Bob by $\sigma_1$ resp.~$\sigma$.

On $\rho$ the messages received by the referee are mutually independent, i.e., $\sigma$ is a product state (Alice's message is a quantum state, and Bob's message is classical. The messages are neither entangled nor correlated with each other).
On $\rho_1$, the messages of Alice and Bob are (usually) not independent. The marginal states of Alice's message under $\rho$ and under $\rho_1$ are equal, and the same holds for Bob. That means that $\sigma$ is the product of the marginal states of $\sigma_1$.

From \bref{m_t2},
\m{\gamma_0>\I[(A(X_0,X),B(X_0))\sim\sigma_{1}]AB
=\KL{\sigma_{1}}{\sigma}
\ge\fr2{\ln 2}||\sigma_{1}-\sigma||_t^2.}
Now, by symmetrically using the same argument we also get
\m{\gamma_0>\I[(A(X_0,X),B(X_0))\sim\sigma_{2}]AB
\ge\fr2{\ln 2} ||\sigma_{2}-\sigma||_t^2 .}

Via the triangle inequality we get that $||\sigma_1-\sigma_2||_t\leq O(\sqrt{\gamma_0})$. Note that we can use the same $t_0$ for both estimates, since most $t_0$ are good.
Again, as $\gamma_0$ can be made arbitrarily small, the protocol $\Cl{P}$ makes an error with probability $1/2-\aso1$ (unless $a\tm b=\asOm n$), which completes our proof.

\theodup{the:oot}{\myTheoOOT}{In any QR protocol for the ``One out of two'' problem, the product of the message lengths must be $\Omega(n)$.}

We note that there is a protocol for the ``One out of two'' problem as described in Section~\ref{s_method_sect}, in which Alice sends a deterministic message and Bob a randomised message, both of length $O(\sqrt n)$, and indeed other trade-offs of the form $ab=O(n)$ are possible by arranging the inputs as non-square matrices.

\section{Lower Bounds}\label{lower_sect}
\subsection{Bounds for \NE}

\begin{theorem}\label{the:QRQ}
Any QRQ protocol for \NE, in which the message lengths of Alice, Bob, Merlin are $a,b,m$ satisfies
$b(a+m)\geq\Omega(n)$.
\end{theorem}

\begin{corollary}
For QRQ protocols, if $ab=o(n)$, then $bm=\Omega(n)$, and indeed $m=\Omega(n/b)\geq\Omega(\min\{n/a,n/b\})$ as claimed in Table 1.
\end{corollary}

\begin{proof}[Proof of Theorem \ref{the:QRQ}]
Take any QRQ $(a,b,m)$-protocol for \NE. We show that this implies a QR $(a+m,b)$-protocol for the ``One out of two'' problem and we are done by Theorem~1.
Instead of sending a quantum proof/message of length $m$ to the referee we let Merlin simply provide Alice with a string $Z$ claimed to be equal to $Y$, Bob's real input. Alice can now provide Referee with her own message (depending on $X$ only) and the proof $\rho_{X,Z}$ that would maximise Merlin's success probability if Bob's input is really $Z$. Clearly the communication is $a+m$ from Alice, and still $b$ from (unchanged) Bob.

The new protocol is as good as the previous at distinguishing $X\neq Y$ and $X=Y$: in both situations its maximum acceptance is the same as in the QRQ protocol.

We claim that the new protocol also solves the ``One out of two'' problem: given inputs $X,Y,Z$ and the promise that $X=Y$ or $Z=Y$ (and $X\neq Z$) then in the first case the protocol will reject, in the second case accept (with high probability).

Hence by Theorem~1 , we get $b(a+m)\geq\Omega(n)$.
\end{proof}

The following corollary is easy by symmetry.
\begin{corollary}
Any RRQ or QRR $(a,b,m)$-protocol for $\NE\ $ satisfies $b(a+m),a(b+m)\geq\Omega(n)$, and hence when $ab=o(n)$ we have $m\geq\Omega(\max\{n/a,n/b\})$.
\end{corollary}

\begin{proof}
The RRQ case is trivial since, for the same protocol, we can use the above argument with the role of Alice and Bob exchanged, and not exchanged (they both send classical messages), leading to lower bounds $b(a+m)\geq\Omega(n)$ and $a(b+m)\geq\Omega(n)$.
For the QRR case this can also be done: now Merlin's message can be sent by either the classical player or the quantum player, since the message itself is not
quantum, and in any case the protocol stays of type QR.
\end{proof}

\subsection{Bounds for \EQ}

\begin{theorem}\label{the:QRQeq}
Any QRQ protocol for \EQ, in which the message lengths of Alice, Bob, Merlin are $a,b,m$ satisfies
$b(a+m)\geq\Omega(n)$.
\end{theorem}

\begin{corollary}
For QRQ protocols, if $ab=o(n)$, then $bm=\Omega(n)$, and indeed $m=\Omega(n/b)\geq\Omega(\min\{n/a,n/b\})$ as claimed in Table 1.
\end{corollary}

\begin{proof}[Proof of Theorem \ref{the:QRQeq}]
In the case of $X=Y$ there is a proof/message from Merlin $\rho_X$ that makes the protocol accept with high probability. Note that Alice knows this proof, since
she knows $X$.

To solve the ``One out of two'' problem again, Alice will simply choose the proof $\rho_X$, and send it together with her original message. Bob, again remains unchanged. If the referee accepts in the original protocol, then he will now announce that $X=Y$, otherwise that $Z=Y$.

Note that when $X=Y$ then the proof $\rho_X$ makes the Referee accept with high probability, and if not, then $Z=Y$ via the promise, and
the referee will reject, because the proof $\rho_Z$ on inputs $X\neq Y$ must lead to rejection with high probability in the original protocol.
\end{proof}

Again by symmetry we get the following corollary:

\begin{corollary}
Any RRQ $(a,b,m)$-protocol for $\EQ$ satisfies $b(a+m),a(b+m)\geq\Omega(n)$, and hence when $ab=o(n)$ we have $m\geq\Omega(\max\{n/a,n/b\})$.
\end{corollary}

We now turn to the case of QRR protocols, which exhibits a fundamentally different trade-off compared to the same case for \NE: while for \NE\  a proof from Merlin of length $n^{2/3}$ requires messages of length $n^{1/3}$ from Alice and Bob, for \EQ\ any sub-linear proof requires $ab=\Omega(n)$, at which point the problem can be solved without Merlin.

\begin{theorem}
Any QRR or RRR protocol for \EQ, in which the message lengths of Alice, Bob, Merlin are $a,b,m$ satisfies
$ab+m\geq\Omega(n)$.
\end{theorem}

\begin{proof}

Given a QRR protocol for \EQ\ first observe that the classical message by Merlin can be assumed to be deterministic. He knows the whole protocol (although not the internal or measurement randomness of Alice, Bob, Referee). For any probability distribution on messages
his winning probability (Merlin's goal is to make the referee accept) is a convex combination of winning probabilities over deterministic messages, and he can as well pick the best of those without losing anything.

Assuming that $m\leq n/3$ (otherwise we are done) we can find a message for Merlin such that the message maximises acceptance for at least $2^{2n/3}$
inputs $X,X$ while still rejecting inputs $X,Y$ with $X\neq Y$ with high probability. Hence we can find a QR protocol that solves a sub-problem of \EQ\ that by renaming is equivalent to \EQ\ on $2n/3$ bit inputs, and we get that $ab=\Omega(n)$.
\end{proof}

In fact we also have a more general statement for any function $f$.
\begin{theorem} For any Boolean function $f$: if we have a QRR $(a,b,m)$-protocol for $f$, then $ab+m\geq \Omega(N(f))$, for the non-deterministic communication complexity $N(f)$ of $f$.
\end{theorem}

\begin{proof}[Proof (sketch)]
Again we may assume that Merlin's message is deterministic. The idea is to fix proof messages $M$ of Merlin, and thus obtain partial functions $f_M$ accepting all 1-inputs for which $M$ is a good proof, while rejecting all 0-inputs. Using a result of Klauck and Podder \cite{KlauckPodder14} we can find a deterministic one-way protocol for $f_M$ of cost $O(ab)$, and putting Merlin's proofs back in gives us a non-deterministic protocol.
\end{proof}

\section{Protocols}\label{upper_sect}
In this section we give two protocols that illustrate the tightness of most of our bounds.

\subsection{Tightness for \NE\ and RRR protocols}

We start with the case of $\NE$. Here RRR, QRR, RRQ protocols all have the same complexity, as shown by our lower bounds and the following upper bound.

\begin{theorem}\label{prot:ne}
There is an RRR $(a+\log n,a+\log n,m+\log a)$-protocol for \NE\  for all $a,m$ such that $am\geq c\cdot n$ for some constant $c$.
\end{theorem}

\begin{proof}

Consider the following protocol.
\itemi{
\item Alice, Bob, Merlin fix an error-correcting code $C:\01^n\to\01^N$, where $N=\asO n$ and $C(x)$ differs from $C(y)$ on at least $N/3$ positions for every $x\ne y$.
Additionally, let $N=am$ for some $a,m\in\NN$.
\item Viewing the code-words of $C$ as $a\times m$ Boolean matrices, Alice chooses a uniformly random $i\in[m]$, and Bob uniformly random $j\in[m]$.
Merlin (if honest) chooses some $k\in[a]$ such that $C(x)$ and $C(y)$ differ on at least $m/3$ positions of the \ord[k] row.
\item Alice and Bob send columns $i$ resp.~$j$ of the encodings of $x,y$ plus the numbers $i,j$.
\item Merlin sends row $k$ of both of the encodings of $x,y$ plus the number $k$.
\item If the \ord[k] entry of the column sent by Alice is not equal to the \ord[i] entry of the row sent by Merlin, then the referee rejects. Similarly, if the \ord[k] entry of the column sent by Bob is not equal to the \ord[j] entry of the row sent by Merlin, reject.
\item Accept, if the rows sent by Merlin for $x,y$ differ in at least $m/3$ positions, otherwise reject.
}

If $x\neq y$ then Merlin can proceed as indicated, and the protocol will accept with certainty.

It remains to show that Merlin can not cheat if $x=y$. Denote by $r_k$ and $s_k$ the two rows sent by Merlin, which coincide in at most $2m/3$ positions. $a_i$ and $\beta_j$ are the columns sent by Alice and Bob.

If Merlin cheats by changing $u$ entries in $r_k$, then he will be caught with probability $u/m$ by the test against Alice's message, similarly he will be caught with probability $v/m$ if he changes $v$ entries in $s_k$. But to pass the last test $u+v\geq m/3$. Hence the total probability with which he will be caught is at least $1/3$.

We thus have an $(a,a,m)$-protocol with $am=O(n)$ that has perfect completeness and soundness error $2/3$. Repetition can improve this to arbitrarily small error.

\end{proof}

\subsection{Tightness for QRQ and RRQ protocols}

In the next section we will show that a quantum message of length $\log n$ can be replaced by a classical message of length $O(\sqrt n\log n)$ by a trusted player, and a quantum message of length $O(\sqrt n\log n)$ by an untrusted player. This is the new primitive Untrusted Quantum State Transfer (UQST).

Furthermore we need the following fact from~\cite{BCWW01_Qu_Fi}: a QQ protocol for \EQ\  (or \NE) needs only communication $\log n+O(1)$ from either player. In this protocol, the players send superpositions over the indices and entries of an error-correcting code for $x,y$, and the swap-test is used by the referee to tell whether $x=y$ or not.

Hence we may replace the $\log n+O(1)$ qubit message from Bob by a $O(\sqrt n\log n)$ bit randomised message from Bob together with a $O(\sqrt n\log n)$ quantum message from Merlin, leading to a QRQ, $(\log n, \sqrt n\log n,\sqrt n\log n)$-protocol. Note that this works for both \EQ\ and \NE.

The same approach works for other values of Bob's message length, leading to a protocol in which Alice sends $O(\log n)$ qubits, Bob $O(b\log b)$ bits, and Merlin $O(n/b\log b)$ qubits.

We would like to stress that such a protocol is impossible in the QRR or the RRQ case: our lower bounds show that for \NE\ Merlin's message needs to be at least $\Omega(n/\log n)$ (qu-)bits long if another player only sends $O(\log n)$ qubits.

We now turn to RRQ protocols. Informally, the idea is to ``de-quantise'' both messages from Alice and Bob as above. This leads to a protocol in which Alice and Bob both send $O(a\log a)$ bits, whereas Merlin sends $O(n/a\cdot\log n)$ qubits (see Section \ref{rrq_protocol_sect}).

Note that the RRQ protocol also works for \NE, however, the protocol in Theorem \ref{prot:ne} is simpler, slightly more efficient, and does not use quantum messages.

\section{Untrusted Quantum State Transfer}\label{uqst_sect}

For the task of untrusted quantum state transfer (UQST) players Alice and Merlin, holding the classical description of a pure quantum state $|\phi\rangle$ on $\log n$ qubits, have to provide messages to the referee, such that the referee can get a single copy of a state $\rho$ that is $\epsilon$-close to $|\phi\rangle$ in the trace distance. Here Merlin can send quantum messages, but is untrusted, whereas trusted Alice can send only classical (randomised) messages. We are interested in the lengths of the messages they have to send.

More formally, in a protocol for UQST Alice and Merlin send messages to the referee.
The referee produces two outputs: a classical bit meaning acceptance or rejection, and a quantum state $\rho$ on $\log n$ qubits.  The protocol is $(\epsilon,\delta)$-successful, if

\begin{enumerate}
\item For every quantum message from Merlin: the probability of the event that $\rho$ satisfies
$||\rho-|\phi\rangle\langle\phi|\,||_t>\epsilon$ and the referee accepts (simultaneously) is at most $\delta$.
The probability is over randomness in the (possibly mixed) quantum state $\rho$ as well as  Alice's random choices and the referee's measurements.
\item there is a message from Merlin such that the referee will accept with probability at least $1-\delta$.
\end{enumerate}

Note that in the discrete version of this problem $|\phi\rangle$ is given as a vector of $n$ floating point numbers with (say) $100\log n$ bits precision each.

\subsection{A protocol}

\begin{theorem}
For any $a\geq10\log(n)/(\epsilon^6\delta^6)$ Untrusted Quantum State Transfer can be implemented with Alice sending $O(a\log (a/(\epsilon\delta)))+O(\log n)$ bits and Merlin $O((n/a)\log n/(\delta^3\epsilon^2))$ qubits.
\end{theorem}

\begin{proof}
We may assume that $a\leq \epsilon n$, because otherwise Alice can send the classical description and Merlin is not needed.

We start by describing the protocol. Let $|\phi\rangle\in\mathbb{C}^n$.
To simplify presentation we first assume that Alice and the referee share randomness (without communicating). We will remove this assumption later. The additional communication cost of removing this shared randomness is at most $O(\log n)$ if the quantum state $|\phi\rangle$  is described by $poly(n)$ classical bits.

Alice and the referee start by choosing a random $a$-dimensional subspace $V$ of $\mathbb{C}^n$. The distribution we use is the normalised Haar measure on $a$-dimensional subspaces (the uniform distribution). Denote
$|\phi_V\rangle=Proj_V |\phi\rangle/||Proj_V|\phi\rangle||$. Alice and the referee agree beforehand on a fixed basis for each subspace of dimension $a$. Alice's message is the classical description of $|\phi_V\rangle$ within precision $\epsilon^{2}\delta^4/100$, i.e., the description of a state $|\psi\rangle$ such that $||\,|\psi\rangle\langle\psi|-|\phi_V\rangle\langle\phi_V|\,||_t\leq\epsilon^{2}\delta^4/100$. To achieve this Alice can simply send each entry of the projected vector as a floating point numbers with $\log a+O(\log(1/(\epsilon\delta)))$ bits precision, i.e., Alice's message has length $O(a\log(a/(\epsilon\delta)))$. We stress that the referee now knows a classical description of $|\psi\rangle$.

We will describe Merlin, the untrusted player, and his behaviour in terms of the situation where he is honest, and later analyse the probability that we catch him when dishonest. The honest Merlin simply sends $m=200n/(a\epsilon^2\delta^3)$ copies of $|\phi\rangle$, i.e., $|\phi\rangle^{\otimes 200n/(a\epsilon^2\delta^3)}$.

Now let us turn to the referee.
He picks a random $i\in[1,...,m]$. This index will be the group of $\log n$ qubits in Merlin's message that he will output (or use otherwise) if the rest of the message passes all tests.

The referee knows the random subspace $V$, and measures each consecutive block (labelled from $\{1,\ldots, m\}$) of $\log n$ qubits in Merlin's message using the observable $V,I-V$, except block $i$ (we abuse notation by identifying the space $V$ with the projector on $V$). Projection on $V$ is considered acceptance of the measurement, and we say that a block {\em survives} the measurement if it accepts.

If none of the measurements results in acceptance then the referee will reject. Otherwise some $k$ blocks $i_1,\ldots, i_k$ will be accepted. After the measurements the referee (for $j\in\{1,\ldots, k\}$) applies the swap-test to the (projected) states he holds on the qubits of the blocks $i_j$ and the state $|\psi\rangle$.
If {\bf any} of the tests fail then the referee rejects. If all tests pass, the referee takes the copy in block $i$ as his output and accepts.
Note that a classical description of $|\psi\rangle$ is known to the referee, so he can use a fresh copy of this state for every test (in fact it is slightly better to measure the blocks by the observable consisting of the projection on $|\psi\rangle$ and the projection on its orthogonal subspace instead).

First we analyse the probability with which an honest Merlin can make the referee accept. The main problem here is the precision in Alice's message. But before that we need to know the number $k$ of copies that survive" the measurement.

\begin{lemma}
Let $|\phi\rangle$ denote any pure state in $n$ dimensions, and $V$ a random $a$-dimensional subspace (under the Haar measure).
Then the probability that $|\phi\rangle$, measured by $V,I-V$, will be accepted is $a/n$.
\end{lemma}

\begin{proof}
Due to symmetry we may assume that $V$ is the space spanned by the first $a$ vectors in the standard basis, and $|\phi\rangle$ is a random vector on the $n-1$-dimensional sphere. The probability of acceptance is the expected squared length of the projection of $|\phi\rangle$ on $V$.
This is $a$ times the expectation of $|\langle\phi|e_1\rangle|^2$, where $e_1$ is any standard basis vector. Due to symmetry the squared projection length onto any basis vector is $1/n$.
\end{proof}

We will later also need information regarding the deviation from the expectation. This is provided by the following main ingredient to the Johnson-Lindenstrauss Lemma, see~\cite{Gupta}.

\begin{fact}\label{fact:jl}
Let $|\phi\rangle$ denote any pure state in $n$ dimensions, and $V$ a random $a$-dimensional subspace (under the Haar measure). Then for the squared length $L$ of $|\phi\rangle$ projected to $V$ and all $0\leq\beta<1/2$ we have:

\[Prob(L\leq (1-\beta)a/n)\leq e^{-a\beta^2/4}.
\]
\[Prob(L\geq (1+\beta)a/n)\leq e^{-a\beta^2/8}.\]

\end{fact}

We have that $E[k]=200/(\delta^3\epsilon^2)$, and the probability that $k<100/(\delta^3\epsilon^2)$ is at most $e^{-25/(\delta^3\epsilon^2)}<\delta/2$ by the Chernoff bound.
If $k>100/(\delta^3\epsilon^2)$ then the referee will keep exactly $k=100/(\delta^3\epsilon^2)$ copies and discard the remaining ones.
Since Merlin is honest, all $k$ projected states are copies of $|\phi_V\rangle$, and we get that
$F^2(\phi_V,\psi)\geq 1-||\,|\psi\rangle\langle\psi|-|\phi_V\rangle\langle\phi_V|\,||_t= 1-\epsilon^{2}\delta^4/100$.
Hence the swap-test will succeed on a copy with probability at least $1/2+(1-\epsilon^{2}\delta^4/100)/2\geq 1-\epsilon^{2}\delta^4/200$.

The probability that at least one of the $k$ swap-tests fails is at most $k\cdot\epsilon^{2}\delta^4/200<\delta/2$.
In total the failure probability is hence at most $\delta$, and on acceptance the resulting quantum state is $|\phi\rangle$, which is exactly as desired.

Now consider a dishonest Merlin, who can send any quantum message $\sigma$ on $m\log n=200(n/a)/(\epsilon^2\delta^3)\cdot\log n$ qubits.
First observe that since the referee measures the $m$ blocks (presumed to be copies of $|\phi\rangle$) separately and consecutively, there is no advantage for Merlin to entangle the blocks. This follows from (inductively applying) the following claim.

\begin{claim}\footnote{To prove this claim consider the density matrix of $\rho_{AB}$ in a product basis, and change it into an block-diagonal matrix by replacing the off-diagonal entries with 0. Each POVM-element acts on a diagonal block only, and hence results are unchanged.}
Let $\rho_{AB}$ denote a bipartite quantum state. Then there is an unentangled state $\sigma_{AB}$, such that the results of all measurements acting on $A$ and $B$ alone are the same for both states.
\end{claim}
Hence Merlin's message $\sigma$ is w.l.o.g.~separable across blocks, i.e., a probability distribution on products of pure states across blocks.
Denote by $\phi$ the density matrix of $|\phi\rangle$. We will show the following:

\begin{lemma} For every pure product state sent by Merlin, i.e., every
$\sigma=\sigma_1\otimes\cdots\otimes\sigma_m$ with all $\sigma_i$ pure, either $E_j ||\sigma_j -\phi||_t\leq \epsilon\delta/2$ or
the referee will reject with probability at least $1-\delta/2$.

\end{lemma}

But if
$E_j ||\sigma_j -\phi||_t\leq \epsilon\delta/2$, then the probability that the output (a random $\sigma_i$) satisfies
$||\sigma_i -\phi||_t\geq \epsilon$ is at most another $\delta/2$ and in total the error probability is no more than $\delta$.

Now suppose that Merlin sends a probability distribution on a product of pure states, i.e., $\sigma=\sum_l p_l \sigma^{(l)}$, where $\sigma^{(l)}=\sigma^{(l)}_1\otimes\cdots\otimes\sigma^{(l)}_m$ and the $p_l$ are probabilities, the $\sigma_j^{(l)}$ pure states. Denote by $a_l$ the acceptance probability on $\sigma^{(l)}$, and $d_l=E_j ||\sigma^{(l)}_j -\phi||_t$. The probability that the average block is far is $f=\sum_{l:d_l\geq\epsilon\delta/2} p_l$, and the probability of acceptance on each $l$ where $d_l\geq\epsilon\delta/2$ is at most $\delta/2$. So the joint probability of acceptance and being far is $f\delta/2\leq\delta/2$, and again with probability $1-\delta$ the output will be good.

We now prove the lemma. In the following $\sigma=\sigma_1\otimes\cdots\otimes\sigma_m$, where all $\sigma_j$ are pure states.

Recall that Merlin does not know $V$, and since $V$ is drawn independently of Merlin's message the expected probability that block $j$ will survive the measurement is exactly $a/n$.
The referee will end up using $k=100/(\epsilon^2\delta^3)$ blocks $i_1,\ldots,i_k$ (and reject if there are fewer blocks that survive the measurement). Also, $i$ (the number of the block that the referee retains) is chosen uniformly random from $1$ to $m$.

Now assume that Merlin cheats significantly. As discussed this means for us
that $E_j ||\sigma_j-\phi||_t>\epsilon\delta/2$. We need to show that the state $\sigma$ leads to rejection with probability at least $1-\delta/2$ in this situation.
We will show that then the swap-tests will make the referee reject with high probability. But before the swap-tests the referee measures $(V,I-V)$, leading to acceptance of $k$ blocks (or more, but those are discarded).
The probability that a block will ``survive'' the measurement is the squared projection length onto $V$. $V$ is not known to Merlin, and for each $\sigma_j$ the expected survival probability is $a/n$. A potential problem here is that Merlin might somehow be able to skew the set of surviving blocks towards those that inside $V$ look like $\phi$, even though they are globally far away from $\phi$. This is impossible because Merlin does not know $V$.

In fact Merlin is not able to influence the set of chosen positions much.
The random experiment of picking the blocks that survive consists of two steps: first picking $V$ to determine the projection lengths of the vectors in each block, and then separately measuring each block, in effect independent coin tosses with the probability given by the squared projection lengths. By Fact~\ref{fact:jl} we have that each projection length is sharply concentrated around $a/n$, and since $a\geq10\log n/(\epsilon^6\delta^6)$ and the union bound {\em all} of the squared projection lengths are between $(1-\epsilon^3\delta^3)a/n$ and $(1+\epsilon^3\delta^3)a/n$ with high probability.
Hence the resulting distribution on blocks chosen is close to the one where each block is chosen independently with probability exactly $a/n$, and we can assume that blocks are chosen as in the latter instead.

We will now show that after the measurement (and resulting projection on $V$ and re-normalisation), states $\sigma_j$ that are far from $|\phi\rangle$ lead to states that are far from $|\phi_V\rangle$.

Write $\sigma_j=|\psi_j\rangle\langle\psi_j|$ and denote $||\sigma_j-\phi||_t=\gamma_j$. Our assumption is that $E_j[\gamma_j]=\gamma\geq\epsilon\delta/2$. We have
\begin{eqnarray*}
1-\gamma_j^2/4&=&1-||\phi-\sigma_j||_t^2/4\\
&\geq &F^2(\phi,\sigma_j)\\
&=&\langle\phi|\psi_j\rangle^2\\
&=&(1-||\, |\phi\rangle-|\psi_j\rangle\,||^2/2)^2\\
&\geq&1- ||\,|\phi\rangle-|\psi_j\rangle\,||^2.
\end{eqnarray*}

Consequently $||\,|\phi\rangle-|\psi_j\rangle\,||^2\geq\gamma_j^2/4$.

Denote $|\psi_{j,V}\rangle=\mbox{Proj}_V |\psi_j\rangle/||\,\mbox{Proj}_V|\psi_j\rangle||$ and $\ell_j=||\,\mbox{Proj}_V|\psi_j\rangle||$
and $\ell=||\,\mbox{Proj}_V||\phi\rangle||$.
Consider the vector $u_j=|\phi\rangle-|\psi_j\rangle$ and $w_j=\mbox{Proj}_V(u_j)$.
By Fact~\ref{fact:jl} $E[||w_j||^2]=a/n\cdot ||u_j||^2 $ and the distribution is tightly concentrated around its mean.
In particular the probability that $||w_j||^2\leq (a/n)\gamma^2_j/8$ is at most $e^{-a/16}\leq n^{-1/(2\epsilon^6\delta^6)}$.
By the union bound we can easily suppress the probability that this happens for any $j$.

We are interested in the distribution of $||\,|\psi_{j,V}\rangle-|\phi_V\rangle\,||^2=
||\,\mbox{Proj}_V|\psi_j\rangle/\ell_j-\mbox{Proj}_V||\phi\rangle/\ell||^2$.
We have already argued that $|\ell^2-a/n|,|\ell_j^2-a/n|\leq\epsilon^3\delta^3 (a/n)$, and hence
$||\,|\psi_{j,V}\rangle-|\phi_V\rangle\,||^2\geq ||w_j||^2\cdot (n/a)-2\epsilon^3\delta^3\geq\gamma_j^2/8-2\epsilon^3\delta^3$ with high probability for all $j$.

This means the projected and re-normalised vectors are still quite far apart in the squared euclidean distance, and hence their fidelity is bounded away from 1, and the swap test will fail with good probability, making the overall test fail with high probability.

In particular we have that for all $j$ the fidelity $F^2(\psi_{j,V},\phi)\leq 1-\gamma_j^2/16+2\epsilon^3\delta^3$, and the \ord[j] swap test succeeds with probability at most $1/2+(1-\gamma_j^2/16+2\epsilon^3\delta^3)/2\leq 1-\gamma^2_j/32+\epsilon^3\delta^3$. Hence the total acceptance probability is at most
$\prod_{j=1}^k (1-\gamma^2_j/32+\epsilon^3\delta^3)\leq(1-\gamma^2/32+\epsilon^3\delta^3)^k\leq(1-\epsilon^2\delta^2/100)^k$ by the arithmetic-geometric mean inequality. This probability is bounded by $\delta/4$ as long as $\delta<1/3$.

So the total probability that either not enough blocks of high distance are chosen or that the swap-tests succeed is at most $\delta/2$.

Finally, let us state the following result by Newman~\cite{newman:random}, which allows us to remove the shared randomness between Alice and the referee at the expense using/sending $O(\log n)$ private random bits in addition to her message.

\begin{fact}
In a randomised protocol using public coin protocol computing a function $f$ on $n$ input
bits for each player needs only $\log n+O(1)$ random bits.
\end{fact}

In fact we first need to make the random choice discrete (by using an $\epsilon$-net), and then apply the above fact. Even though our protocol does not exactly compute a function (for a given input there are many correct outputs), the Chernoff bound arguments in Newman's proof go through in our situation.

\end{proof}

\subsection{A lower bound}

We now observe that our protocol is optimal within lower order terms.

\begin{theorem}
Any protocol for UQST in which Alice sends $u$ bits and Merlin $v$ qubits on states on $\log n$ qubits must satisfy $u(v+\log n)\geq\Omega(n)$.
\end{theorem}

\begin{proof}
Assume there is a protocol for UQST with message lengths $u,v$.
As shown in Theorem~3 any QRQ protocol for \EQ\  must satisfy $b(a+m)\geq\Omega(n)$. There is a QQ $(a,b)$-protocol for \EQ\ (see~\cite{BCWW01_Qu_Fi}) with $a=b=\log n+O(1)$, which  can be turned into a QRQ protocol with message lengths $(\log n+O(1), O(u),O(v))$ by applying UQST to the message of Bob, making Bob randomised and introducing Merlin. Hence $u(v\log n)\geq\Omega(n)$.
\end{proof}

\section{Tightness for RRQ Protocols}\label{rrq_protocol_sect}

We sketch a protocol for the RRQ situation now. We assume for now that Alice and the referee, as well as Bob and the referee share a public coin (Alice and Bob do not). Alice and the referee pick a random $a$-dimensional subspace, and Alice sends the projection of the quantum fingerprint for $x$ onto the space to the referee. Bob and the referee do the same. The projection is sent using $O(a\log a)$ bits. Honest Merlin sends $O(n/a)$ copies of the quantum fingerprint of $x$ to the referee, in tensor product. The referee chooses half of these copies and measures them with the subspace agreed on with Alice, and the other half with Bob's. We expect that at least a constant number of copies `survive' the measurement in both cases. The referee then measures them to check them against the transmitted states. If any of these measurements fail, he rejects. Clearly, with similar arguments as above, in the case of $x=y$ the honest Merlin will convince the referee with high probability. If $x\neq y$, again Merlin cannot improve his chance of convincing the referee by sending something other than a product state. The referee chooses random copies to do the projection measurements, so Merlin can not predict which copy is supposed to be the fingerprint of $x$ or of $y$.

\section{An RRR Protocol for Disjointness}\label{disj_sect}

In this section we make an observation regarding the Disjointness problem, which, while trivially hard in most versions of the SMP model, becomes more interesting with the addition of the prover Merlin. Specifically, we show that there is a non-deterministic SMP protocol with total cost $a+b+m=O(n^{2/3}\log n)$, by adapting an $O(\sqrt{n}\log n)$-cost protocol of Aaronson and Wigderson \cite{AW09} for the Disjointness and Inner Product problems in a related model. The increase in the communication is due to the need to avoid introducing shared randomness between the players.

For the lower bound we only have the $\Omega(\sqrt{n})$ bound which holds trivially because $MA(DISJ)=\Omega(\sqrt{n})$ and every non-deterministic SMP protocol can be converted into an MA-protocol by having Alice and Bob each play the role of the referee for one another. However, given the weakness of the remaining interaction once a proof is fixed, it would not be surprising for the true lower bound to be higher.

\begin{theorem} For any $a$ such that $\sqrt{n}\leq a\leq n^{3/4}$, there is a non-deterministic SMP protocol for $DISJ$ with cost $(a\log n,a\log n,\left(\frac{n}{a}\right)^{2}\log n)$ (up to constants).
\end{theorem}
\begin{proof}
As in the protocol of Aaronson and Wigderson \cite{AW09}, we divide the set $[n]$ into $n^{\alpha}$ blocks of size $n^{1-\alpha}$, with $\alpha\geq1/2$, and write the inputs $x$ and $y$ respectively as functions $a,b:[n^{\alpha}]\times [n^{1-\alpha}]\rightarrow\{0,1\}$ such that $x_{n^{1-\alpha}(i-1)+j}=a(i,j)$ and $y_{n^{1-\alpha}(i-1)+j}=b(i,j)$. The problem then is to accept with high probability if $\sum_{i\in[n^{\alpha}]}\sum_{j\in[n^{1-\alpha}]}a(i,j)b(i,j)=0$ and reject with high probability otherwise. Choosing a prime $q\in (n,2n]$, the functions $a$ and $b$ can then be extended uniquely to two-variable polynomials $\tilde{a},\tilde{b}:\mathbb{F}_{q}^{2}\rightarrow \mathbb{F}_{q}$ which have total degree $\leq O(n^{\alpha})$. We also define the $O(n^{\alpha})$-degree polynomial $s(i)=\sum_{j\in[n^{1-\alpha}]}\tilde{a}(i,j)\tilde{b}(i,j)\mod q$.

For the protocol, Merlin sends a candidate polynomial $s'$ of degree equal to that of $s$ to the referee. Now if Alice and Bob could each send for some randomly chosen value $r$ the values $\tilde{a}(r,j)$ and $\tilde{b}(r,j)$ (respectively) for all $j\in[n^{1-\alpha}]$, as well as the value $r$, then the referee could compute $s(r)$ and test it against $s'(r)$: If the size of the set from which $r$ is chosen is sufficiently large, and $s\neq s'$, then with sufficiently high probability $s(r)$ would differ from $s'(r)$, and the referee could reject (otherwise he would assume that $s'=s$ and use $s'$ to decide whether to accept or reject). Unfortunately, the random choice of $r$, which is necessary in order to keep it hidden from Merlin, also keeps it hidden from the player who is not making the choice. Thus, if Bob chooses $r$, the only way for Alice to know which block of values $\tilde{a}(i,j)$ to send, is by either communicating or having shared randomness with Bob.

This issue can of course be overcome by having each player send the blocks of $n^{1-\alpha}$ values $\tilde{a}(r,j)$ (resp. $\tilde{b}(r,j)$), together with $r$, for sufficiently many values of $r$ as to have a collision with high probability. In order to do this efficiently, we need to make the set from which $r$ is chosen as small as possible, while keeping it large enough to be able to distinguish $s$ from $s'$ reliably if they differ. For the Schwartz-Zippel Lemma it suffices for the set $S$ from which $r$ is chosen to be of size, say, $10\deg(s)=O(n^{\alpha})$, while Alice and Bob each choose $100\sqrt{|S|}=O(n^{\alpha/2})$ values of $r$.

What remains is to analyse the cost of this protocol. Merlin uses $O(n^{\alpha}\log n)$ bits to send the coefficients of the polynomial $s'$. Alice and Bob send $O(n^{\alpha/2})$ blocks of $O(n^{1-\alpha})$ values in $[q]$ each, which costs $O(n^{1-\alpha/2}\log n)$ bits of communication. Setting $a=n^{1-\alpha/2}$, we get the statement of the theorem. For $\alpha=2/3$ the message length of each party is the same up to a constant.
\end{proof}

{}
\subsection{Discussion}
{}

The main purpose of this work was to address what we viewed as remaining gaps in the understanding of one of the most basic communication primitives -- the Equality function.
We have developed a new technique for analysing \EQ.
Compared to previously used methods, it has the following advantages:\itemi{
\item It seems to be more widely applicable, as it allowed us to show \asOm{\sqrt n} lower bound for both \EQ\ and its negation \NE\ in the \e{non-deterministic} version of quantum-classical SMP, where Merlin is also quantum.
This is the strongest known version of SMP where the complexity of both \EQ\ and \NE\ remains high, and the previously known lower bound techniques seemed to be insufficient for showing that.\fn
{The best lower bound that we were able to prove using combinations of known techniques is $\tilde{\Omega}(n^{1/3})$.}
\item It provides a unified view upon the complexity of \EQ\ in all the versions of SMP mentioned above.
Moreover, it simplifies some previously known bounds (even for the best-understood case of \EQ\ in the randomised SMP, the new proof is, arguably, the simplest known).
}

Additionally, we showed that there is an RRR protocol for Disjointness with communication $O(n^{2/3}\log n)$ from Alice, Bob, and Merlin.
We conjecture this to be tight up to the log-factor.

{}

{}

\end{document}